\documentclass[11pt,a4paper]{article}
\pdfoutput=1
\usepackage{amssymb,amsthm,amsmath,graphicx,clrscode,xspace,microtype}
\usepackage{hyperref}\hypersetup{colorlinks=false, pdfborder={0 0 0}}
\usepackage[numbers,sort&compress]{natbib}
\newtheorem{theorem}{Theorem}[section]\newtheorem{lemma}[theorem]{Lemma}\newtheorem{corollary}[theorem]{Corollary}

\newcommand{\e}{\emph}\renewcommand{\cal}[1]{\ensuremath{\mathcal{#1}}\xspace}
\renewcommand{\l}{\ensuremath{\lambda}\xspace}
\renewcommand{\sp}{\ensuremath{\mathsf{SP}}\xspace}\newcommand{\m}{\ensuremath{\mathsf{m}}\xspace}
\newcommand{\st}[1]{\ensuremath{\mathsf{St}({#1})}\xspace}
\newcommand{\opt}{\textsf{opt}\xspace}\newcommand{\apx}{\textsf{apx}\xspace}

\title{Improved approximations for robust mincut and shortest path}
\author{Valentin Polishchuk \qquad Mikko Sysikaski \\{\small{Helsinki Institute for Information Technology}} \\ {{\small\texttt{firstname.lastname@cs.helsinki.fi}}}}
\date{}
\begin{document}          %\thispagestyle{empty}\pagestyle{empty}
\maketitle\maketitle\begin{abstract}In two-stage robust optimization the solution to a problem is built in two stages: In the first stage a partial, not necessarily feasible, solution is exhibited. Then the adversary chooses the ``worst'' scenario from a predefined set of scenarios. In the second stage, the first-stage solution is extended to become feasible for the chosen scenario. The costs at the second stage are larger than at the first one, and the objective is to minimize the total cost paid in the two stages.

We give a 2-approximation algorithm for the robust mincut problem and a $(\gamma+2)$-approximation for the robust shortest path problem, where $\gamma$ is the approximation ratio for the Steiner tree. This improves the factors $1+\sqrt2$ and $2(\gamma+2)$ from {{[Golovin, Goyal and Ravi. Pay today for a rainy day: Improved approximation algorithms for demand-robust min-cut and shortest path problems. \e{STACS 2006}]}}. In addition, our solution for robust shortest path is simpler and more efficient than the earlier ones; this is achieved by a more direct algorithm and analysis, not using some of the standard demand-robust optimization techniques.
\end{abstract}
{\bf Keywords: } Approximation algorithms, Demand-robust optimization
\section{Introduction}
The general setting in a two-stage optimization problem is as follows: There is a set of \e{demands} (aka \e{scenarios}), one of which has to be satisfied tomorrow. It is not until tomorrow that it is revealed which demand must be satisfied. A demand is satisfied by buying a set of \e{resources}. Thus, one possibility to satisfy the tomorrow's demand is to wait until tomorrow, know the scenario, and buy a corresponding set of resources. However, the resources are cheaper today than tomorrow, by an ``inflation'' factor $\l>1$. Hence it makes sense to buy some resources already today, i.e., at the \e{first stage}, even without knowing the tomorrow's scenario. (Say, if $\l=\infty$, the resources bought today should better satisfy all demands.) Then tomorrow, upon revealing the requested demand, only some additional, \e{second-stage}, resources have to be bought at the higher price.

The tomorrow's demand is chosen by an adversary. The adversary knows the resources bought at the first stage. He also knows the algorithm that you will use for buying second-stage resources. The adversary chooses the scenario so that your second-stage cost is as large possible (the adversary is omnipotent, and can solve an NP-hard problem for that, if necessary). Your objective is thus to minimize the \e{maximum, worst-case} cost paid in the two stages. Because of such hedging against the worst demand, this type of two-stage robust optimization is called \e{demand-robust}.

\subsection*{Related work}In \e{stochastic optimization} (see, e.g., \cite{stoch1,stoch2}) the objective is to minimize the \e{expected} cost paid over the two stages. \e{Universal approximations} \cite{universal1,universal2} may in a sense be viewed as \e{one-stage} robust solutions. The \e{demand-robust optimization} as studied in this paper was introduced by Dhamdhere et al.\ in~\cite{focs05}. Several techniques have proved to be viable in the field:

\noindent\e{Minimal feasible solutions. } Dhamdhere et al.~\cite{focs05} showed that there always exists an approximate first-stage solution which is a minimal feasible solution for a subset of scenarios. Restricting oneself to such solutions makes one loose at most a factor of 2 in comparison with the unrestricted case. Since the pioneering paper \cite{focs05}, the minimal-solution idea has been extensively used in the design of approximation algorithms for two-stage robust optimization problems.

\noindent\e{LP rounding. } IP formulations of optimization problems often extend directly to stochastic and demand-robust versions; rounding the LP relaxation solution is one way to obtain an approximation.

\noindent\e{Thresholded $\alpha$-approximations. } A common approach to demand-robust optimization is as follows: Suppose you are shooting for an $\alpha$-approximation. Guess the second-stage cost $C_{II}^*$ of the optimal solution (often the number of relevant $C_{II}^*$s is small; if worse comes to worst, go through "all possible" $C_{II}^*$s approximately with repeated doubling -- or more precisely, with repeated $(1+\varepsilon)$-ing). In the first stage, satisfy all high-cost demands -- those each of which is more expensive than $\alpha C_{II}^*$ to satisfy. Then in the second stage you are guaranteed to pay at most $\alpha C_{II}^*$ -- which is within factor $\alpha$ of optimal second-stage cost. Finally, argue that your first-stage solution is also within $\alpha$ times the first-stage cost of the optimum -- for the overall approximation guarantee of $\alpha$.
A very general treatment of the thresholded covering algorithms is presented in a recent paper~\cite{icalp10}.% the thresholded algorithms were applied also earlier -- virtually in all papers on demand-robust optimization. Strictly speaking, if we ignore the mincut on trees in \cite{focs05}, \cite{stacs06} seem to be the first to use thresholding.
%
%A common approach to demand-robust optimization includes scrolling through \e{all} possible values of the second-stage cost incurred by the \e{optimal} solution; with an explicit list of scenarios, it is often possible to also explicitly list all such values. For the case of exponential scenarios, the scrolling is done approximately -- an efficient and simple general framework including such scrolling is presented by Gupta, Nagarajan and Ravi in~\cite{icalp10}.
%stacs06 runs in SP + |T||V| StT but they may use Melhorn for apx StT (at the increase in runtime)
% (to the best of our knowledge, efficiency of approximation algorithms for demand-robust optimization was not addressed previously) -- not entirely true, stacs06 states the runtime for mincut.
\subsection*{Our contributions}In Section~\ref{sec:mincut} we present a thresholded 2-approximation for the robust mincut problem. This improves the (also thresholded) $(1+\sqrt2)$-approximation from \cite{stacs06}. The improved approximation guarantee is due to a refined analysis using, similarly to \cite{stacs06}, laminarity of mincuts (the Gomory-Hu mincuts tree).

In Section~\ref{sec:sp} we give a $(\gamma+2)$-approximation algorithm for the robust shortest path problem, where $\gamma$ is the Steiner tree approximation ratio. This improves the $2(\gamma+2)$-approximation from \cite{stacs06} (the techniques in \cite{arxivicalp10} potentially imply a 4.25-approximation). The algorithm and its analysis are very simple. In particular, unlike \cite{stacs06} we do not restrict ourselves to minimal feasible solutions and do not use the thresholding. Avoiding the guessing of the second-stage cost of the optimum makes our algorithm more efficient (by at least a linear factor) than that of~\cite{stacs06}.
%We showed that a $\gamma$-approximation for the Steiner tree implies an $(\gamma+2)$-approximation for the robust shortest path. Clearly, the demand-robust shortest path problem is no easier than the Steiner tree: for $\l=\infty$ the optimal solution is a Steiner tree on $T$. Thus, a $\gamma$-approximation for the demand-robust path implies a $\gamma$-approximation to the Steiner tree. It is interesting whether this could be strengthened, i.e., whether the robust path problem is strictly harder than the Steiner tree: Is it true that a $\gamma$-approximation for the demand-robust path implies an $f(\gamma)$-approximation for the Steiner tree, for some function $f(\gamma)<\gamma$?
%
%Incidentally, our approximation ratio is exactly half of that in \cite{stacs06}, who lose a factor of 2 in their analysis due to employing the structural result that restricting the first-stage solution to be a tree increases the approximation ratio by the factor of at most 2 (see \cite[Lemma~3.1]{stacs06} and \cite{focs05}). Our first-stage solution is also a tree; however, we avoid using the result by doing a more straightforward (and simpler) analysis.

\section{Demand-robust mincut}\label{sec:mincut}
%\paragraph{Problem formulation}
In the demand-robust mincut problem the input is a (positively) weighted undirected graph $G=(V,E,C)$ with $C$ representing the \e{capacities} of edges, the \e{root} vertex $r\in V$, and a set $T\subseteq V\setminus r$ of \e{terminals}. For a terminal $t\in T$ let $\m(t)$ denote the minimum $r\-t$ cut (if the mincut is not unique, take $\m(t)$ to be the cut that cuts out from $r$ a maximal set of vertices); similarly, for a set $S\subseteq T$ of terminals, $\m(S)$ is the minimum $r\-S$ cut. We use $C(t),C(S)$ to denote the capacities $C(\m(t)),C(\m(S))$ of the mincuts $\m(t),\m(S)$.
For a subset $E'\subseteq E$ of edges let $\m_{E'}(t)$ be the minimum $r\-t$ cut in $G$ with weights of edges in $E'$ set to 0; let $C_{E'}(t)$ denote the capacity of the mincut $\m_{E'}(t)$.

A feasible solution to the robust mincut problem is an arbitrary set $E_I\subseteq E$ of edges. The cost of the solution is
\[C(E_I)+\l\cdot\max_{t\in T}C_{E_I}(t)\]
where \l is the inflation factor.

The edges $E_I$ of the solution are called the \e{first-stage} edges and the cost $C(E_I)$ --- \e{first-stage} cost; the edges $\m_{E_I}(t)$ are the \e{second-stage} edges for terminal $t$ and the cost $\max_{t\in T}C_{E_I}(t)$ is the \e{second-stage} cost. The objective is to find $E_I$ minimizing the two-stage cost (with the second-stage cost inflated by \l).
\paragraph{Notation}For a set $P\subseteq V$ of vertices let $\partial P$ denote the \e{boundary} of $P$ --- the set of edges that have exactly one endpoint in $P$. We use $E_I^*$ to denote the optimal solution.
\subsection{Mincuts laminarity}Let $G^*$ be $G$ with the capacities of edges in $E_I^*$ set to 0. For a terminal $t\in T$, let $Q^*_t\subseteq V\setminus r$ denote the $t$-side of the cut $\m_{E_I^*}(t)$ --- the vertices reachable from $t$ after the edges $E_I^*$ and $\m_{E_I^*}(t)$ are removed (the asterisk emphasizes that $Q^*$ is the $t$-side of the mincut in $G^*$, not in the original $G$). It is known (e.g., can be seen from the Gomory-Hu tree \cite[Section~8.6]{coBook}) that these $t$-sides do not properly intersect --- $\forall u,v\in T$ either $Q^*_u\cap Q^*_v=\emptyset$ or $Q^*_u\subseteq Q^*_v$ or $Q^*_v\subseteq Q^*_u$. In other words, for any subset $S\subseteq T$ of terminals the $t$-sides of the terminals in $S$ form a laminar family $\cal{F}_S^*=\{Q^*_t:t\in S\} $ of sets.

Let $F_S^*\subseteq\cal{F}_S^*$ be the basic (inclusion-maximal) sets in the family $\cal{F}_S^*$; assume that all sets in $F_S^*$ are unique (note that in principle we could have $Q^*_u=Q^*_v=Q^*$ for $u,v\in S,u\neq v$, with $Q^*$ not being a proper subset of any other set in $\cal{F}_S^*$ --- in this case only one of $Q^*_u,Q^*_v$ is included in $F_S^*$). Call the terminals $B_S^*=\{b\in S:Q^*_b\in F_S^*\}$ the \e{basic} terminals of~$S$.
\subsection{Thresholded $\alpha$-approximation}The thresholded covering paradigm applied to the robust mincut problem works as follows: Imagine that we know the cost $C_{II}^*=\max_{t\in T}C_{E_I^*}(t)$ that the optimum pays at the second stage. To obtain an $\alpha$-approximate solution, cut out the set $U=\{t\in T:C(t)>\alpha C_{II}^*\}$ of "expensive" terminals in the first stage. That is, the output of the algorithm is $\m(U)$.% and its cost is \begin{multline*}c(\m(U))+\l\max_{t\in T}c(\m_{\m(U)}(t))=\\=c(\m(U))+\l\max_{t\in T\setminus U}c(\m_{\m(U)}(t))\le\\\le c(\m(U))+\alpha\l C\end{multline*}

Assuming the terminals in $T=(t_1,\dots,t_{|T|})$ are ordered in non-increasing order of mincut ($C(t_i)\ge C(t_{i+1})$), for any $C_{II}^*$ we have $U=\{t_1,t_2,\dots,t_j\}$ for some $j=j(C_{II}^*)\in\{0,1,2,\dots,|T|\}$. Hence there are only $|T|+1$ different possible sets $U$ for all possible $C_{II}^*\ge0$ --- so try all the possibilities and choose the best; this way the non-determinism in $C$ is reduced just to the non-determinism w.r.t.~$j$.

For the approximation ratio analysis assume that the algorithm is run with the "right" guess of $U$ corresponding to the right choice of~$C_{II}^*$.
By definition of $U$, the thresholded algorithm pays at most $\alpha C_{II}^*$ in the second stage. The tricky part is to bound the cost, $C(U)$, of the first stage.% as a function of $\alpha$. Below we outline the
\paragraph{The analysis of \cite{stacs06}} Golovin, Goyal and Ravi \cite{stacs06} %observe, and let $\partial Q^*_b$ denote the boundary of $Q^*_b$; that $(\alpha-1)\,C(\partial Q^*_b\setminus E_I^*)\le C(\partial Q^*_b\cap E_I^*)$ for any basic terminal $b\in B^*_U$ of $U$, and plug the inequality into
use the following estimate of the first-stage cost of the thresholded $\alpha$-approximation:
\begin{align}\label{stacsBound}C(U)\quad\le\quad C(\bigcup_{b\in B^*_U} \partial Q^*_b)\quad\le \\ %\sum_{b\in B^*_U}C(\partial Q^*_b)\quad\le \\
\notag\le \quad C(E_I^*)+\frac1{\alpha-1}\sum_{b\in B^*_U}C(\partial Q^*_b\cap E_I^*)\end{align}
To bound the last sum, \cite{stacs06} cleverly use the fact that due to pairwise-disjointness of the basic sets $F^*_U$, every edge $e\in E_I^*$ appears at most twice in the sum; thus $C(U)\le(1+\frac2{\alpha-1})C(E_I^*)$, and the overall approximation ratio of the thresholded $\alpha$-approximation algorithm is $\max(1+\frac2{\alpha-1},\alpha)$, minimized by $\alpha=1+\sqrt2$ --- the final approximation ratio of~\cite{stacs06}.
\subsection*{Using $\alpha=2$ is enough}Our algorithm is just the thresholded 2-approximation, i.e., the output of our algorithm is the minimum $r\-U$ cut where $U=\{t\in T:C(t)>2C_{II}^*\}$; as usual, for the analysis we assume that $C_{II}^*$ (or, equivalently, $U$) was guessed correctly. The second-stage cost of our solution is at most $2C_{II}^*$. In what follows we prove the bound $C(U)\le2C(E_I^*)$ on the fist-stage cost of our algorithm.

To show $C(U)\le2C(E_I^*)$, instead of a correct but too generous bound (\ref{stacsBound}) of \cite{stacs06} on $C(U)$, we use a tighter estimate
\begin{equation*}\label{eq:us}C(U) \quad \le \quad C\left(\partial\bigcup_{b\in B_U^*} Q_b^*\right)\end{equation*}
The correctness of the estimate follows from the same argument as in \cite{stacs06}: every terminal of $U$ belongs to at least one of the sets $Q_b^*$, and none of the sets $Q_b^*$ contains $r$; thus the boundary of the union is an $r\-U$ cut.

To prove
\[C\left(\partial\bigcup_{b\in B_U^*}Q_b^*\right)\quad\le\quad2C(E_I^*)\]
we argue that
\begin{equation}\label{bigeq}C\left(\partial\bigcup_{b\in B_U^*}Q_b^*\right)\quad\le\quad2C\left(\Delta^*\left(\bigcup_{b\in B_U^*}Q_b^*\right)\right)\end{equation}
where $\Delta^*(P)=((P\times P)\cap E_I^*) \cup (\partial P\cap E_I^*)$ denotes the edges from $E_I^*$ that have \e{at least one} endpoint in a set $P\subseteq V$ of vertices; clearly, the right-hand-side of (\ref{bigeq}) is at most $2C(E_I^*)$.

Number the terminals in $B_U^*$ arbitrarily: $B_U^*=(b_1,b_2,\dots)$. For $k=0,1,\dots,|B_U^*|$ define $B^*_k=Q_{b_1}^*\cup Q_{b_2}^*\cup\dots\cup Q_{b_k}^*$. The inequality (\ref{bigeq}) follows from the next lemma:
\begin{lemma}$\forall k=0,1,\dots,|B_U^*|,\quad C\left(\partial B^*_k\right)\le2C\left(\Delta^*\left(B^*_k\right)\right)$\end{lemma}
\begin{proof}By induction on $k$. The base is trivial: $0=C(\partial\emptyset)\le2C(\Delta^*(\emptyset))=0$.

Let $X^*=\partial(B^*_{k-1},Q^*_{b_k})\cap E_I^*,X=\partial(B^*_{k-1},Q^*_{b_k})\setminus E_I^*,Y^*=(\partial Q^*_{b_k}\setminus (X^*\cup X))\cap E_I^*,Y=(\partial Q^*_{b_k}\setminus (X^*\cup X))\setminus E_I^*$ (Fig.~\ref{induction}). %, and let $x^*=C(X*),x=C(X),y^*=C(Y\cap E_I^*),y=C(Y\setminus E_I^*)$ . %By (\ref{high}), $x^*+y^*\ge x+y\ge y-x$, or $y^*\ge y-(x^*+x)$, and
\begin{figure}\centering\includegraphics{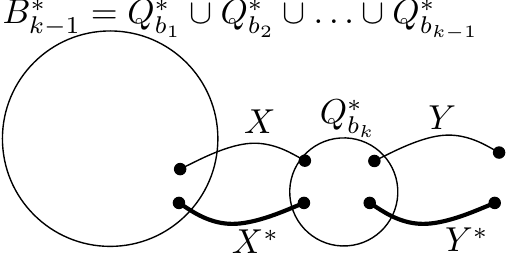}\caption{$E_I^*$ is bold. $B^*_k=B^*_{k-1}\cup Q^*_{b_k}$.}\label{induction}\end{figure}

Because $b_k$ belongs to the set of high-cost terminals $U$, the optimal solution $E_I^*$ must "help" $b_k$ by at least half (using the terminology from \cite{icalp10}, $b_k$ is "low"):
\begin{equation}\label{high}C(X^*)+C(Y^*)\ge C(X)+C(Y)\end{equation}
Indeed, since $\partial Q^*_{b_k}$ is an $r\-b_k$ cut, $C(b_k)\le C(\partial Q^*_{b_k})=C(X^*)+C(Y^*)+C(X)+C(Y)$, and since $b_k\in U$, $C(b_k)\ge2C_{II}^*\ge2C_{E_I^*}(b_k)=2(C(X)+C(Y))$, from where (\ref{high}) follows.

Using $C(X)\ge0$ we rewrite (\ref{high}) as
\begin{align}\label{high_alt}\notag C(Y^*)\ge\frac{C(Y^*)}2+\frac{C(Y)-C(X^*)+C(X)}2\ge\\
\ge\frac{C(Y^*)+C(Y)-C(X^*)-C(X)}2\end{align}

We have (see Fig~\ref{induction}):
\begin{equation}\label{eq:delta}C(\Delta^*(B_k))\ge C(\Delta^*(B^*_{k-1}))+C(Y^*)\end{equation}
\begin{equation}\label{eq:cut}C(\partial B^*_k)=C(\partial B^*_{k-1})+C(Y^*)+C(Y)-C(X^*)-C(X)\end{equation}

By the inductive hypothesis,
\begin{equation}\label{eq:hypothesis}C(\Delta^*(B_{k-1}))\ge\frac{C(\partial B^*_{k-1})}2\end{equation}

Putting (\ref{high_alt}), (\ref{eq:delta}), (\ref{eq:cut}), (\ref{eq:hypothesis}) together we obtain
\begin{multline*}C(\Delta^*(B_k))\ge \frac{C(\partial B^*_{k-1})}2 + \frac{C(Y^*)+C(Y)-C(X^*)-C(X)}2 = \frac{C(\partial B^*_k)}{2}
\end{multline*}
\end{proof}
Overall, we have that the first-stage cost of our solution is at most $2C(E_I^*)$, and the second-stage cost is at most $2C_{II}^*$:
\begin{theorem}There is a polynomial-time algorithm which gives a 2-approximation for the robust mincut problem.\end{theorem}

\section{Demand-robust shortest path}\label{sec:sp}
The input to the demand-robust shortest path is the same as to the demand-robust mincut problem: graph $G=(V,E,w)$ with $w$ representing the \e{lengths} of edges, root vertex $r\in V$, and a set $T\subseteq V\setminus r$ of terminals. A solution is a set $E_I\subseteq E$ of edges. The cost of the solution is
\[w(E_I)+\l\cdot\max_{t\in T}w(\sp_{E_I}(t))\]
where $\sp_{E_I}(t)$ is the shortest $r\-t$ path in $G$ with weights of edges in $E_I$ set to 0. The objective is to find $E_I$ minimizing the two-stage cost.
\subsection{Algorithm}
Our solution is a Steiner tree on a subset $S\subseteq T$ of \l terminals. The set $S$ is built incrementally, starting from $r$, and repeatedly adding a farthest (with ties broken arbitrarily) terminal, until gathering \l of them:
\begin{codebox}
%\Procname{$\proc{SimpleRobSP}(G=(V,E,w),r,T,\l)$} Can we refer to the alg's name later?
\li $S\leftarrow\{r\}$
\li \While $|S| \le \min\left(\l,|T|\right)$ \label{size}
\li \Do $S\gets S\cup\mathrm{arg}\max\limits_{u\in T\setminus S}\sp(u,S)$ \qquad\Comment Add farthest terminal \End \label{sp}% Strictly speaking, argmax can be a whole set, but lets not worry.
\li \Return $E_I\gets$ approximate Steiner tree on $S$ \label{st}
%\li Second Stage Solution: Shortest path from the demand to $S$
\end{codebox}
\subsection{Approximation ratio}We now analyze the approximation guarantee of the algorithm. Let $f$ be the distance from $S$ to the terminal added in the last iteration of the {\bf while} loop (line~\ref{sp}). Because $|S|$ was growing from iteration to iteration, we have that at \e{any} iteration the distance from $S$ to the farthest terminal was at least $f$. Thus, the distance between any two vertices in $S$ is at least $f$. Hence, $\st{S}\ge|S|f/2$ where \st{S} is the weight of the minimum Steiner tree on $S$ (to see this, go twice around the tree --- you traveled 2\st{S}, spending at least $f$ traveling between any two vertices in~$S$). That is,
\[\label{f}f\le\frac{2\st{S}}{|S|}\]

Because we always add farthest terminal to $S$, at the completion of the algorithm the distance from any terminal to $S$ is at most $f$. Hence,
\[w(E_{II})\le f\]
where $E_{II}$ are the edges that we buy at the second stage.

Let $E^*_I$ be the optimal solution. Let $E^*_{II}(t)$ be the edges that the optimal solution buys at the second stage if the demand is $t\in T$. Let $E^*_{II}$ be the edges that the optimum buys in the worst case: $w(E^*_{II})=\max_tw(E^*_{II}(t))$. Then $E^*_I\cup_{t\in S\setminus r}E^*_{II}(t)$ is a connected graph that spans $S$. Hence
\[w(E^*_I)+\sum_{t\in S\setminus r}w(E^*_{II}(t))\ge\st{S}\]
and
\[w(E^*_I)+(|S|-1)w(E^*_{II})\ge\st{S}\]

We consider the cases $\l\ge|T|$ and $\l<|T|$ separately:

If $\l\ge|T|$, then $S=T\cup r$, and
\begin{multline*}
\opt = w(E^*_I)+\l w(E^*_{II}) \quad\ge\quad w(E^*_I) + |T|w(E^*_{II}) \quad= \\
= \quad w(E^*_I) + (|S|-1)w(E^*_{II}) \quad\ge\quad \st{S} \quad\ge\quad w(E_I)/\gamma \quad=\quad \apx/\gamma
\end{multline*}
where \opt is the optimal cost, \apx is what we pay, and $\gamma$ is the approximation factor for the Steiner tree. That is, $\apx\le\gamma\,\opt$.

If $\l<|T|$, then
\begin{multline*}
\apx \quad=\quad w(E_I) + \l w(E_{II}) \quad\le\quad \gamma\st{S}+\l f \quad\le \\
\le\quad \left(\gamma+\frac{2\l}{|S|}\right)\st{S} \quad\le\quad \left(\gamma+\frac{2\l}{|S|}\right)(w(E^*_I)+(|S|-1)w(E^*_{II})) \quad\le \\
\le\quad (\gamma+2)(w(E^*_I)+ \l w(E^*_{II})) \quad=\quad (\gamma+2)\,\opt
\end{multline*}
because $\l\le|S|\le\l+1$.
\begin{theorem}If for some class of graphs there is a $\gamma$-approximation for Steiner tree, then for that class of graphs there is a $(\gamma+2)$-approximation for robust shortest path.\end{theorem}
For general graphs, the best $\gamma=1.39$ is due to Byrka et al.~\cite{stoc10}:%Robins and Zelikovsky \cite{robinsZelikovsky}:
\begin{corollary}There is a polynomial-time algorithm which gives a 3.39-approximation for the robust shortest path problem.\end{corollary}
\subsection{Running time}As far as the efficiency of approximating the robust shortest path is concerned, the best bound that can be given on the running time of the algorithm of \cite{stacs06} is $O(|T||V||E|)$. The (multiplicative) overhead of $O(T||V|)$ is due to guessing $|T||V|$ possible values for the second-stage cost of the optimal solution. Then for each guess the algorithm of \cite{stacs06} builds an approximate Steiner tree, which must take $\Omega(|E|)$ time.

Because our algorithm avoids the guessing, we can achieve the running time of $O(\min(\l,|T|)(|E|+|V|\log|V|))$ at the expense of increasing the approximation ratio to 4. For that, in line~\ref{st} we use the $O(|E|+|V|\log|V|)$-time 2-approximation algorithm of Mehlhorn~\cite{mehlhorn} or Floren~\cite{floren}. Then our algorithm's running time is dominated by finding the farthest terminals in line~\ref{sp}.
%stacs06 runs in SP + |T||V|*StT but they may use Melhorn for apx StT (at the increase in apx ratio)
% (to the best of our knowledge, efficiency of approximation algorithms for demand-robust optimization was not addressed previously) -- not true, stacs06 states the runtime for mincut.
\begin{corollary}There is an $O(\min(\l,|T|)(|E|+|V|\log|V|))$-time algorithm which gives a 4-approximation for the robust shortest path problem.\end{corollary}
\bibliographystyle{abbrv}\bibliography{robust}
\end{document}